\documentclass[10pt,conference,a4paper]{IEEEtran}
\IEEEoverridecommandlockouts
\usepackage{CJK}
\usepackage{cite}
\usepackage{float}
\usepackage{graphicx}
\usepackage{amsmath}
\usepackage{times}
\usepackage{graphicx}
\usepackage{latexsym}
\usepackage{graphicx}
\usepackage{bm}
\usepackage{amssymb}
\usepackage[center]{caption2}
\usepackage{stfloats}
\usepackage{cases}
\usepackage{array}
\usepackage{setspace}
\usepackage{fancyhdr}
\usepackage{color}
\usepackage{subfigure}
\usepackage{citesort}
\usepackage{multirow}
\usepackage{amsmath,amsthm}
\usepackage{cases}
\usepackage[noend]{algpseudocode}
\usepackage{algorithmicx,algorithm}
\usepackage{mathtools}

\theoremstyle{remark}

\def\bm#1{\mbox{\boldmath $#1$}}

\newtheorem{lemma}{Lemma}

\newtheorem{proposition}{Proposition}
\newtheorem{remark}{Remark}
\renewcommand{\baselinestretch}{0.97}

\renewcommand{\baselinestretch}{0.97}

\newcaptionstyle{mystyle2}{%
\captionlabel $.$ \, \captiontext \par} \captionstyle{mystyle2}

\begin{document}
\title{Robust Design for IRS-Aided Communication Systems with User Location Uncertainty}
\author{\IEEEauthorblockN{Xiaoling Hu, Caijun Zhong, Mohamed-Slim Alouini, and Zhaoyang Zhang}
\thanks{X. Hu, C. Zhong and Z. Zhang are with the College of Information Science and Electronic Engineering, Zhejiang University, Hangzhou, China.(email: caijunzhong@zju.edu.cn).}
\thanks{M.-S. Alouini is with the Department of Computer, Electrical and Mathematical Science and Engineering, King Abdullah University of Science and Technology, Thuwal 23955, Saudi Arabia.}
}
\maketitle

\begin{abstract}
In this paper, we propose a robust design framework for IRS-aided communication systems in the presence of user location uncertainty. By jointly designing the transmit beamforming vector at the BS and phase shifts at the IRS, we aim to minimize the transmit power subject to the worse-case quality of service (QoS) constraint, i.e., ensuring the user rate is above a threshold for all possible user location error realizations.  With unit-modulus, this problem is not convex. The location uncertainty in the QoS constraint further increases the difficulty of solving this problem. By utilizing techniques of Taylor expansion, S-Procedure and semidefinite relaxation (SDP), we transform this problem into a sequence of semidefinite programming (SDP) sub-problems. Simulation results show that the proposed robust algorithm substantially outperforms the non-robust algorithm proposed in the literature, in terms of probability of reaching the required QoS target.
\end{abstract}

\begin{IEEEkeywords}
IRS, robust beamforming, SDP.
\end{IEEEkeywords}

\title{Robust design for an IRS-aided system under user location  uncertainty}

\section{Introduction}
Intelligent reflecting surface (IRS) has recently emerged as a promising method to tackle the coverage problem of dead-zone users \cite{slim,wu2019towards}. The IRS is a meta-surface consisting of a large number of low-cost, passive reflecting elements, each of which can independently reflect the incident signal with adjustable phase shifts. Through proper design of the phase shifts at the IRS \cite{wu2019intelligent,multicast}, the signal received at the user can be significantly enhanced.

To facilitate the design of phase shifts, channel state information (CSI) is compulsory. As such, plenty of works have been devoted to the estimation of IRS-aided channels, including the channel between the BS and IRS, as well as the reflection channel between the IRS and user.
In general, the proposed methods can be categorized into two types. The first type is to estimate the cascaded  channels\cite{zheng2019intelligent,he2019cascaded}. However, with a large number of reflecting elements, the training overhead would become prohibitive. The second type is to directly estimate two channels separately \cite{taha2019enabling}, assuming that the IRS has both reflection mode and receive mode. However, this requires a large number of receive radio frequency (RF) chains (equal to the number of reflecting elements), which would significantly increase the hardware cost as well as the power consumption of the IRS.

To tackle the above issues, one possible way is to only estimate the angles of arrivals (AOAs) and angles of departures (DOAs)\cite{hu1,han2019large,X.Hu1}. Because the locations of the IRS and BS remain fixed, the channel between the BS and IRS varies very slowly and can be accurately estimated by computing the AOAs and AODs \cite{zhou2019robust}. Then, exploiting the user location information provided, for example, by global positioning system (GPS), the angular information of the reflection channel from the IRS to the user can be obtained. However, due to user mobility and the precision of GPS, the user location information may not be accurate, resulting in imperfect angular information.

Motivated by this, this letter proposes a robust design framework in the presence of user location uncertainty.\footnote{  Unlike \cite{zhou2019robust} and \cite{robust1}, which adopt conventional error model, the current work first proposes a location-based channel estimation scheme, and then presents a practical channel error model related to location uncertainty. In addition, the method to tackle the optimization problem is also different from \cite{zhou2019robust} and \cite{robust1}.} Specifically, considering a bounded spherical error model, the transmit beamforming vector and phase shifts are jointly designed to minimize the transmit power, subject to the minimum achievable rate constraint for all possible user location errors. To tackle the non-convex worse-case rate constraint, the second-order Taylor expansion is used to approximate the achievable rate, and the S-Procedure is used to convert the resultant semi-definite constraint into a matrix inequality. Then, applying the semi-definite relaxation (SDR) method, an alternating algorithm is designed, which optimizes the transmit beamforming vector and phase shifts iteratively by solving a sequence of semi-definite programming (SDP) sub-problems.
Simulation results show that the proposed robust algorithm can guarantee the target rate
regardless of the user location errors, which substantially outperforms the non-robust algorithm proposed in \cite{wu2019intelligent}.

\section{System Model}
We consider an IRS-aided system as illustrated in Fig.\ref{f0}, where one BS with $N$ antennas communicates with a single-antenna user, which is assisted by an IRS with $M$ reflecting elements. Like most works on IRS, e.g.,\cite{wu2019intelligent,zheng2019intelligent}, we assume that the IRS operates in the far-field regime. Furthermore, we assume that direct link  between the BS and the user does not exist, due to blockage or unfavorable propagation environments. Both the BS and the IRS are equipped with uniform rectangular arrays (URAs) with the size of $N_y \times N_z$ and $M_y \times M_z$ respectively, where $N_y$ ($N_z$) and $M_y$ ($M_z$) denote the numbers of  BS antennas and reflecting elements along the $y$ ($z$) axis, respectively.

\begin{figure}[t]
\includegraphics[width=3in]{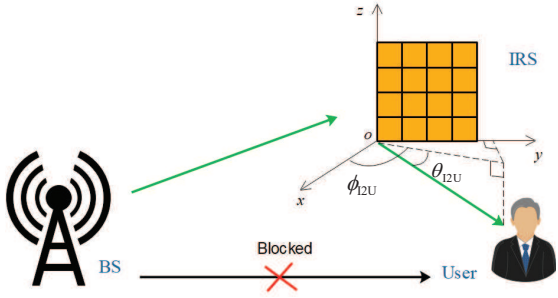}
\centering
\caption{  Model of the IRS-aided communication system. }
\label{f0}
\end{figure}

\subsection{Downlink Transmission}
During the downlink data transmission phase, the BS transmits the signal ${\bf x}=  {\bf w}  s $, where ${\bf w} $ is the beamforming vector  and $s $ is the symbol for the user, satisfying $\mathbb{E}\left\{ |s|^2\right\}=1$.

Then, the signal received  at the   user is given by
\begin{align}
    y = {\bf g}_{\text{I2U} }^T { \bm \Theta } {\bf G}_{\text{B2I} }    {\bf w}  s +n,
\end{align}
where ${\bf g}_{\text{I2U} } \in \mathbb{C}^{M \times 1}$ denotes the channel from the IRS to the user, ${\bf G}_{\text{B2I}} \in \mathbb{C}^{M \times N} $ is the channel between the BS and the IRS, and $n$ is the additive white Gaussian noise (AWGN) which follows the circularly symmetric complex Gaussian distribution with zero mean and variance $\sigma_0^2$.
The phase shift matrix of the IRS is given by  ${\bm \Theta}= {\text {diag}} \left( {\bm \xi}\right) \in \mathbb{C}^{M \times M}$
with the phase shift  beam $ {\bm \xi}={[ e^{j\vartheta_{1}},..., e^{j\vartheta_{n}},...,  e^{j\vartheta_{M}}]}^T \in \mathbb{C}^{M \times 1}$.

\subsection{Channel Model}
{  We consider a narrowband millimeter-wave (mmWave) system, and adopt the narrowband geometric channel model.} As such, the channel from the IRS to the BS can be expressed as
\begin{align} \label{G}
   {\bf G}_{\text{B2I}}\!=\! \sum\limits_{l=1}^{D} \beta_{ l} {\bf b} (\vartheta_{\text{z-B2Ia},l},\vartheta_{\text{y-B2Ia},l} )
{\bf a}^T ( \vartheta_{\text{z-B2I},l},\vartheta_{\text{y-B2I},l} ) ,
\end{align}
where $D$ is the number of paths, $\beta_{ l} $ is the channel coefficient of the $l$-th path,   ${\bf a}$  and ${\bf b}$ are the array response vectors of the BS and IRS respectively.
 The effective angles  of departure (AODs) of the $l$-th path, i.e., the phase differences between two adjacent antennas  along $z$ and $y$ axes,  are given by
$
 \vartheta_{\text{z-B2I},l}=-\frac{2 \pi d_{\text{BS}}}{\lambda} \sin \theta_{\text{B2I},l}$ and
 $\vartheta_{\text{y-B2I},l}=-\frac{2 \pi d_{\text{BS}}}{\lambda} \cos \theta_{\text{B2I},l} \sin \phi_{\text{B2I},l}$, respectively, where $d_{\text{BS}}$ is the  distance between two adjacent BS antennas,
$\lambda$ is the carrier wavelength, $ \theta_{\text{B2I},l}$ and $\phi_{\text{B2I},l}$ are the elevation and azimuth  AODs, respectively. Similarly, the two effective angles of arrival (AOAs) can be written as
 $\vartheta_{\text{z-B2Ia},l}=\frac{2 \pi d_{\text{IRS}}}{\lambda} \cos \theta_{\text{B2Ia},l} \cos \phi_{\text{B2Ia},l}$ and
 $\vartheta_{\text{y-B2Ia}}= \frac{2 \pi d_{\text{IRS}} }{\lambda} \cos \theta_{\text{B2Ia},l} \sin \phi_{\text{B2Ia},l}$, respectively,
 where $d_{\text{IRS}}$ is the  distance between two adjacent reflecting elements,
 $ \theta_{\text{B2Ia},l}$ and $\phi_{\text{B2Ia},l}$ are the elevation and azimuth AOAs, respectively. We further assume that $d_{\text{BS}}=d_{\text{IRS}}=\frac{\lambda}{2}$.

The $s$-th element of $ {\bf a} \left( \vartheta_{\text {z}}, \vartheta_{\text {y}} \right) \in \mathbb{C}^{N \times 1}$ and the $i$-th element of
${\bf b} \left( \vartheta_{\text {z}}, \vartheta_{\text {y}} \right) \in \mathbb{C}^{M \times 1}$ are respectively given by
 \begin{align}
 & \left[ {\bf a}\left( \vartheta_{\text {z}}, \vartheta_{\text {y}} \right)\right]_{s}
 =e^{j\pi\left\{
(s_m-1) \vartheta_{\text {z}}+(s_n-1)  \vartheta_{\text {y}}\right\}},\\
& \left[ {\bf b}\left( \vartheta_{\text {z}}, \vartheta_{\text {y}} \right)\right]_{i}
 =e^{j\pi\left\{
(i_m-1) \vartheta_{\text {z}}+(i_n-1)  \vartheta_{\text{y}}\right\}},\\
& s_m=s-(\lceil  {s}/{N_z} \rceil-1) N_z, \
 s_n=\lceil {s}/{N_z} \rceil,\\
& i_m=i-(\lceil  {i}/{M_z} \rceil-1) M_z, \
 i_n=\lceil {i}/{M_z} \rceil,
 \end{align}
 where $j=\sqrt{-1}$, and we use $[{\bf z}]_i$ to denote the $i$-th entry of a vector ${\bf z}$.

Since the IRS is usually deployed near the user, a LOS channel model is assumed to model the reflection channel between the IRS and the user. Specifically, the channel from the IRS to the user is given by
\begin{align}
   {\bf g}_{\text{I2U} }^T= \alpha_{\text{I2U} } {\bf b}^T \left(\vartheta_{\text{z-I2U} },\vartheta_{\text{y-I2U} } \right),
\end{align}
 where  $\alpha_{\text{I2U} }$ is the channel coefficient,  the two effective AODs from the IRS to the user  $\vartheta_{\text{z-I2U}}$ and
 $\vartheta_{\text{y-I2U}}$ are respectively defined as \cite{array}
 \begin{align}
 & \vartheta_{\text{z-I2U}}\!=\!-\frac{2   d_{\text{IRS}}}{\lambda} \sin \theta_{\text{I2U}}\!=\!- \sin \theta_{\text{I2U}}, \label{v1}\\
 & \vartheta_{\text{y-I2U}}\!=\!-\frac{2  d_{\text{IRS}}}{\lambda} \cos \theta_{\text{I2U}} \sin \phi_{\text{I2U}}\!=-\!  \cos \theta_{\text{I2U}} \sin \phi_{\text{I2U}} ,\label{v2}
 \end{align}
 where $\theta_{\text{I2U}}$ and $\phi_{\text{I2U}}$ are  respectively the elevation and azimuth AODs, as shown in Fig.~\ref{f0}.

\subsection{Performance Measure}
Assuming Gaussian signaling, the achievable rate of the system can be expressed as
\begin{align} \label{R}
R=\log_2\left( 1+ {{\left| {\bf g}_{\text{I2U} }^T { \bm \Theta } {\bf G}_{\text{B2I} } {\bf w} \right|}^2}/{\sigma_0^2 }\right).
\end{align}

\section{Location-Based Channel Estimation}
In the IRS-aided communication system, due to the fixed location of the IRS, the channel from the BS to the IRS, i.e., ${\bf G}_\text{B2I}$, usually remains constant over a long period, and thus can be accurately estimated. Therefore, we assume that ${\bf G}_\text{B2I}$ is perfectly known. In contrast, due to user mobility, the reflection channel ${\bf g}_{\text{I2U}}$ varies over the time, hence should be estimated. As such, we propose to exploit user location information to estimate the reflection channel ${\bf g}_{\text{I2U} }$.

The elevation and azimuth AODs ($\theta_{\text{I2U}}$ and $\phi_{\text{I2U}}$) have the following relationship with the locations of the IRS and  user:
\begin{align}
   y_{\text{U}}\!-\! y_{\text{I}}\!= \!d_{\text{I2U}} \cos\theta_{\text{I2U}} \sin \phi_{\text{I2U}},  \
  z_{\text{U}}\!-\! z_{\text{I}}= d_{\text{I2U}} \sin\theta_{\text{I2U}}, \label{v4}
\end{align}
where  $d_{\text{I2U}}$ is the  distance between the IRS and the user, $(x_{\text{I}}, y_{\text{I}}, z_{\text{I}})$ and $(x_{\text{U}}, y_{\text{U}}, z_{\text{U}})$ are locations of the IRS and the user  respectively.
 Substituting  (\ref{v4}) into (\ref{v1}) and (\ref{v2}), we obtain the relationship  between the effective AODs ($\vartheta_{\text{z-I2U}}$ and $\vartheta_{\text{y-I2U}}$) and locations of the IRS and the user:
 \begin{align}
    \vartheta_{\text{y-I2U}}=\left( y_{\text{I}}- y_{\text{U}}\right)/{d_{\text{I2U}}},  \ \
    \vartheta_{\text{z-I2U}}=\left( z_{\text{I}}- z_{\text{U}}\right)/{d_{\text{I2U}}} .
 \end{align}

In general, user location information obtained from the GPS is imperfect. Let  $(\hat x_{\text{U} },\hat  y_{\text{U} },\hat z_{\text{U} })$ denote the  estimated location of the user and ${\bf \Delta}\triangleq [\Delta x_{\text{U} },\Delta y_{\text{U} },\Delta z_{\text{U} }]^T$ denote the estimation error. The estimated effective AODs from the IRS to the user are given by
\begin{align}
\hat\vartheta_{\text{y-I2U} } &= (y_{\text{I}}-\hat y_{\text{U} })/{\hat d_{\text{I2U} }},  \
\hat\vartheta_{\text{z-I2U} } =(z_{\text{I}}-\hat z_{\text{U} })/{\hat d_{\text{I2U} }}\label{AOD2}.
\end{align}

\begin{proposition} \label{p1}
The   effective AODs from the  IRS  to the  user can be approximately decomposed as
\begin{align}
  \vartheta_{\text{z-I2U} } =
   \hat\vartheta_{\text{z-I2U} } +\epsilon_{\text{z-I2U} }, \
  \vartheta_{\text{y-I2U} } =
   \hat\vartheta_{\text{y-I2U} } +\epsilon_{\text{y-I2U} },\label{E5}
\end{align}
where
\begin{align}
& \epsilon_{\text{z-I2U} }
  \!=  \! \frac{ \left( \hat\vartheta^2_{\text{z-I2U} }-1\right) \Delta z_{\text{U}  } \! +  \!\hat\vartheta_{\text{z-I2U} } \hat\vartheta_{\text{y-I2U} } \Delta y_{\text{U}  } \!+ \!\hat\vartheta_{\text{z-I2U} }\hat\vartheta_{\text{x-I2U} } \Delta x_{\text{U}  } }{\hat d_{\text{I2U} }}, \nonumber \\
 & \epsilon_{\text{y-I2U} }
   \!=  \! \frac{ \left( \hat\vartheta^2_{\text{y-I2U} } \!- \!1\right) \Delta y_{\text{U} }  \!+ \! \hat\vartheta_{\text{y-I2U} } \hat\vartheta_{\text{z-I2U} } \Delta z_{\text{U} } \!+ \!\hat\vartheta_{\text{y-I2U} }\hat\vartheta_{\text{x-I2U} } \Delta x_{\text{U} } }{\hat d_{\text{I2U} }},\nonumber
\end{align}
where $( x_{\text{U} }, y_{\text{U} }, z_{\text{U} })$  denotes the accurate location of the user, $\hat \vartheta_{\text{x-I2U} }  \triangleq \frac{x_{\text{I} }-\hat x_{\text{U}  }}{\hat d_{\text{I2U}  }}$,
$\Delta x_{\text{U} }=x_{\text{U}  }-\hat x_{\text{U}   }$, $\Delta y_{\text{U}   }=y_{\text{U}   }-\hat y_{\text{U}   }$ and $\Delta z_{\text{U}   }=z_{\text{U}   }-\hat z_{\text{U}   }$  are location errors along $x$, $y$ and $z$ axes, respectively.
\end{proposition}
\begin{proof}
Starting from   (\ref{AOD2}), we can obtain the desired result.
\end{proof}

Invoking the results given by Proposition \ref{p1}, the reflection channel can be expressed as
\begin{align} \label{ge}
{\bf g}_{\text{I2U} }= \hat{\bf g}_{\text{I2U} } \odot {\bf e}_{\text{I2U} },
\end{align}
where $\odot$ stands for Hadamard product,  $\hat{\bf g}_{\text{I2U} }=  \alpha_\text{I2U} {\bf b}  (\hat\vartheta_{\text{z-I2U} },\hat\vartheta_{\text{y-I2U} }  )$ is the estimated channel, and ${\bf e}_{\text{I2U} }$ is the estimation error with
$
\left[ {\bf e}_{\text{I2U} } \right]_i=
    e^{j\pi {\bf f}_{  i}^T {\bm \Delta}} \label{error1}
$,
where   ${\bf f}_{  i}=[ a_{  i_mi_n}, b_{  i_mi_n}, c_{  i_mi_n}]^T$ with
\begin{align}
& a_{  i_mi_n}\triangleq (i_m-1) \frac{\hat\vartheta_{\text{z-I2U} }\hat\vartheta_{\text{x-I2U} }}{\hat d_{\text{I2U} }}
 + (i_n-1) \frac{\hat\vartheta_{\text{y-I2U} }\hat\vartheta_{\text{x-I2U} }}{\hat d_{\text{I2U} }},\\
& b_{  i_mi_n}\triangleq (i_m-1) \frac{\hat\vartheta_{\text{z-I2U} }\hat\vartheta_{\text{y-I2U} }}{\hat d_{\text{I2U} }}
 + (i_n-1) \frac{\hat\vartheta^2_{\text{y-I2U} }-1 }{\hat d_{\text{I2U} }},
 \end{align}
 \begin{align}
 & c_{  i_mi_n}\triangleq (i_m-1) \frac{ \hat\vartheta^2_{\text{z-I2U} }-1 }{\hat d_{\text{I2U} }}
 + (i_n-1) \frac{\hat\vartheta_{\text{y-I2U} }
 \hat\vartheta_{\text{z-I2U} } }{\hat d_{\text{I2U} }},\\
& i_m=i-(\lceil \frac{i}{M_z} \rceil-1) M_z, \
 i_n=\lceil \frac{i}{M_z} \rceil.
\end{align}

\section{Robust Beamforming Design}
Assuming that the location error is bounded by a sphere with the radius $\Upsilon$, i.e.,$\|{\bf \Delta}\|^2 \le \Upsilon^2$,  we aim to minimize the total transmit power through joint design of beamforming vector ${\bf w}$ and phase shift vector $\bm{\xi}$ under the worst-case QoS constraint, namely, the achievable rate should be above a threshold $r$. Mathematically, the worst-case robust design problem can be formulated as
\begin{align} \label{op1}
  &  \min\limits_{\left\{ {\bf w},\ {\bm \xi} \right\}}
 \left\|{\bf{w}} \right\|^2 \\
&\operatorname{s.t.}
            R \ge r,    \forall \|{\bf \Delta}\|^2 \le \Upsilon^2,
           \  |\xi_i|=1, i=1,\ldots,M. \nonumber
\end{align}

Substituting (\ref{R}) and (\ref{ge}) into (\ref{op1}), we have
\begin{align} \label{op2}
    &\min\limits_{\left\{ {\bf w},\ {\bm \xi} \right\}}
  \left\|{\bf{w}} \right\|^2 \\
 &\operatorname{s.t.}  \begin{array}[t]{lll}
          { | (\hat{\bf g}_{\text{I2U} } \odot {\bf e}_{\text{I2U} })^T { \bm \Theta } {\bf G}_{\text{B2I} } {\bf w}  |}^2 \! \ge \!  (2^r-1 )\sigma_0^2  ,    \forall \|{\bf \Delta}\|^2 \! \le \!  \Upsilon^2,  \\
            |\xi_i|=1, i=1,\ldots,M.
           \end{array}\nonumber
\end{align}

\subsection{Problem Transformation}
Denote ${\bf d} \triangleq
\text{diag}(\hat{\bf g}_{\text{I2U} }) { \bm \Theta } {\bf G}_{\text{B2I} }   {\bf w}$ and $\bar{\bf \Delta} \triangleq \frac{{\bf \Delta}}{\hat d_{\text{I2U} }} $.  The robust constraint in (\ref{op2})   is reformulated as
\begin{align} \label{c1}
   {\left|{\bf e}_{\text{I2U} }^T {\bf d}  \right|}^2
   \ge \left(2^{r}-1\right) \sigma_0^2
   , \forall \|\bar{\bf \Delta}\|^2 \le \frac{\Upsilon^2}{\hat{d}_{\text{I2U} }^2}.
\end{align}

Constraint (\ref{c1}) is a non-convex constraint involving infinitely many inequality constraints due to
the continuity of the location uncertainty set.
To handle the infinite inequalities, we give an approximation of  the left hand side of (\ref{c1}),  which is shown in the following proposition.
\begin{proposition} \label{p2}
The left hand side   of (\ref{c1}) can be approximated as
$
{\left|{\bf e}_{\text{I2U} }^T {\bf d}  \right|}^2
\approx
Q  +2{\bm{\phi}}^{T}  \bar{\bm \Delta} +  \bar{\bm \Delta}^T {\bm{\Phi}} \bar{\bm \Delta},
$
 where $Q   ={\left| {\bf 1}^T {\bf d}  \right|}^2$, ${\bm \phi}  =
     \sum\limits_{m=1}^{M}
     \sum\limits_{n=1}^{M} [{\bf d} ]_m [{\bf d} ]_n^*
      \left(\bar{\bf f}_{  m}- \bar{\bf f}_{  n} \right)$,
$
 \left[{\bm \Phi} \right]_{sl}=
  \sum\limits_{m=1}^{M}
     \sum\limits_{n=1}^{M} [{\bf d} ]_m [{\bf d} ]_n^*
        \left[\bar{\bf f}_{  m}- \bar{\bf f}_{  n} \right]_{s} \left[\bar{\bf f}_{  m}- \bar{\bf f}_{  n} \right]_{l}
$  with $ \left[{\bm \Phi} \right]_{sl}$ representing the entry in the $s$-th row and the $l$-th column of ${\bm \Phi}$ and
$\bar{\bf f}_{  m} \triangleq j\pi \hat d_{\text{I2U} } {\bf f}_{  m}$.
 \end{proposition}
 \begin{proof}
      By applying second-order Taylor expansion, we can obtain the desired result.
 \end{proof}

  Based on Proposition \ref{p2},  (\ref{c1}) can be rewritten as
  \begin{align} \label{c11}
  Q  +2{\bm{\phi}}^{T}  \bar{\bm \Delta} +  \bar{\bm \Delta}^T {\bm{\Phi}} \bar{\bm \Delta}
   \ge \left(2^{r}-1\right) \sigma_0^2
   , \forall \|\bar{\bf \Delta}\|^2 \le \frac{\Upsilon^2}{\hat{d}_{\text{I2U} }^2}.
  \end{align}
 Then, we leverage the following lemma to convert constraint (\ref{c11}) into an equivalent form
involving one matrix inequality.

\begin{lemma} \label{L1}
(General S-Procedure) Consider the quadratic matrix inequality (QMI)\cite{luo2004multivariate}:
\begin{align}
   h({\bf X})={\bf X}^{H} {\bf A} {\bf X}+ 2\text{Re} \left({\bf B}^{H} {\bf X} \right)+{\bf C} \succeq {\bf 0}, \\
    \forall {\bf X} \in \left\{ {\bf Y}| \text{tr}\left({\bf D}{\bf Y}{\bf Y}^H \right) \le 1, {\bf D} \succeq {\bf 0} \right\}, \nonumber
\end{align}
where $\text{tr}\left(\cdot\right)$ represents  the  trace, ${\bf A},{\bf D}  \in {\mathbb{H}}^{m}$ with ${\mathbb{H}}^{m}$ being the set of $m \times m$ Hermitain matrices, ${\bf X},{\bf B} \in {\mathbb{C}}^{m \times n}$ and ${\bf C}\in {\mathbb{H}}^{n}$. This QMI holds if there exists $\mu\ge0$ such that
\begin{align}
    \begin{bmatrix}
{\bf C} & {\bf B}^H  \\
{\bf B} & {\bf A}
\end{bmatrix}-\mu \begin{bmatrix}
{\bf I}_{n} & {\bf 0} \\
{\bf 0} & -{\bf D}
\end{bmatrix} \succeq {\bf 0},
\end{align}
provided that there exists a point $\bar{\bf X}$ such that $h(\bar{\bf X}) \succeq {\bf 0}$.
\end{lemma}

By applying Lemma \ref{L1}, the constraint (\ref{c11}) is transformed into
\begin{align} \label{c111}
   \begin{bmatrix}
 Q -\left(2^{r }-1\right) \sigma_0^2-\mu  & {\bm{\phi}} ^{T}  \\
{\bm{\phi}}  & {\bm{\Phi}} +\mu  \frac{\hat{d}_{\text{I2U} }^{2} }{\Upsilon^2} {\bf I}_3
\end{bmatrix} \succeq {\bf 0},
\end{align}
where ${  \mu}  \ge 0$ is a slack variable.

The transformed constraint (\ref{c111}) involves only one matrix inequality constraint, which is more amenable for algorithm design compared to the infinitely many constraints in the original
constraint (\ref{c1}). However, the resulting optimization problem is still not jointly convex with respect to
${\bf{w}}$ and ${\bm \xi}$.
 Therefore, we adopt an alternating optimization  method to optimize ${\bf{w}}$ and ${\bm \xi}$ iteratively.

 \subsection{Optimization of ${\bf w}$}
 Specifically, for given phase shift beam ${\bm \xi}$, the sub-problem of problem (\ref{op2}) corresponding to the beamforming vector ${\bf w}$ is formulated as
 \begin{align} \label{sub_op1}
  \begin{array}{ll}
    \min\limits_{\left\{ {\bf w}, { \mu }\right\}}
    \left\|{\bf{w}} \right\|^2, \
 \text{s.t.}  \begin{array}[t]{lll}
               (\ref{c111}), \mu\ge 0.
           \end{array}
  \end{array}
\end{align}

However, the constraint (\ref{c111}) is non-convex.
To proceed, we recast the optimization problem as a rank-constrained SDP problem. By applying the change of variables $\bar{\bf W}={\bf w} {\bf w}^{H}$, problem (\ref{sub_op1}) can be rewritten as
 \begin{align}
    & \min\limits_{\left\{ \bar{\bf W} \in {\mathbb{S}}_{+}^{N},\ { \mu} \right\}}
    \text{tr}\left(\bar{\bf W} \right)    \\
 & \operatorname{s.t.}  \begin{array}[t]{llll}
                \begin{bmatrix}
 \bar{Q} -\left(2^{r}-1\right) \sigma_0^2-\mu & \bar{\bm{\phi}}^{T}  \\
\bar{\bm{\phi}}  & \bar{\bm{\Phi}} +\mu \frac{\hat{d}_{\text{I2U}}^{2} }{\Upsilon^2} {\bf I}_3
\end{bmatrix} \succeq {\bf 0}, \\
           \text{rank}\left( \bar{\bf W}\right)=1, \mu \ge 0,
  \end{array}\nonumber
\end{align}
where ${\mathbb{S}}_{+}^{N}$ denotes the set of  $N \times N$ positive semidefinite  Hermitain matrices,
 $\bar{Q}  ={\bf 1}_{M}^{T} {\bf T}     \bar{\bf W}
 {\bf T} ^H {\bf 1}_{M}$,
$\left[ \bar{\bm \phi} \right]_{q}=\text{tr}\left({\bf D}_{q} {\bf T}  \bar{\bf W} {\bf T}^{H} \right)$,
$\left[\bar{\bm \Phi}  \right]_{sl}=
 \text{tr}\left({\bf A}_{sl} {\bf T} \bar{\bf W} {\bf T} ^{H} \right)$,
where ${\bf 1}_M$ denotes a $M \times1 $ vector whose elements all equal to 1, ${\bf T}    =\text{diag}(\hat{\bf g}_{\text{I2U}}) { \bm \Theta } {\bf G}_{\text{B2I} }$, ${\bf D}_{   q} \in {\mathbb{R}}^{M \times M}$ and ${\bf A}_{   sl}  \in {\mathbb{R}}^{M \times M}$ with  entries in the $m$-th row and the $n$-th column  given by
$[{\bf D}_{   q}]_{mn}= \left[\bar{\bf f}_{   m}- \bar{\bf f}_{   n} \right]_{q}$ and $ [{\bf A}_{   sl}]_{mn}=\left[\bar{\bf f}_{   m}- \bar{\bf f}_{   n} \right]_{s} \left[\bar{\bf f}_{   m}- \bar{\bf f}_{   n} \right]_{l}$.

 As such, the only remaining non-convexity
of problem (14) is due to the rank constraint. Generally, solving such a rank-constrained
problem is known to be NP-hard. To overcome this issue, we adopt the SDR technique and
drop the rank constraint.

 \begin{align} \label{sub_op2}
   &  \min\limits_{\left\{ \bar{\bf W} \in {\mathbb{S}}_{+}^{N},\ { \mu} \right\}}
    \text{tr}\left(\bar{\bf W} \right)    \\
 & \operatorname{s.t.} \begin{array}[t]{llll}
                \begin{bmatrix}
 \bar{Q} -\left(2^{r}-1\right) \sigma_0^2-\mu & \bar{\bm{\phi}}^{T}  \\
\bar{\bm{\phi}}  & \bar{\bm{\Phi}} +\mu \frac{\hat{d}_{\text{I2U}}^{2} }{\Upsilon^2} {\bf I}_3
\end{bmatrix} \succeq {\bf 0}, \
            \mu \ge 0,
  \end{array} \nonumber
\end{align}

Therefore, the resulting problem becomes a convex  SDP problem, and thus can be efficiently solved by standard convex program solvers such as CVX.
While there is no
guarantee that the solution obtained by SDR satisfies the rank constraint,
the Gaussian randomization can be  used to obtain a feasible solution to
problem (\ref{sub_op1}) based on the higher-rank solution obtained by solving (\ref{sub_op2}).

 \subsection{Optimization of ${\bm \xi}$}
 For given beamforming vector ${\bf w}$, the sub-problem of problem (\ref{op2}) corresponding to the phase shift beam ${\bm \xi}$  becomes  a feasibility-check problem. To further improve  optimization performance, we  introduce a slack variable ${ v}$, which is interpreted as the “SINR residual” of the  user. Hence, the feasibility-check problem of $\bm \xi$ is formulated as follows
 \begin{align}  \label{sub_xi1}
  &  \max\limits_{\left\{  {\bm \xi}, { v}, { \mu} \right\}}
 v,  \\
 & \operatorname{s.t.}
           \text{Modified-}(\ref{c111}), {v}\ge 0, { \mu}\ge 0,
            |\xi_i|=1, i=1,\ldots,M,\nonumber
\end{align}
where  the  Modified-(\ref{c111}) is obtained from (\ref{c111})  by replacing $ (2^{r}-1){\sigma}_{0}^{2} $ with $(2^{r}-1){\sigma}_{0}^{2}+v$.

To address the non-convexity of both the constraint (\ref{c111}) and the unit-modulus constraint, we apply change of variables $\bm{\Xi}={\bm \xi}{\bm \xi}^H$. Hence, the problem (\ref{sub_xi1}) can be
transformed into a rank-constrained  (SDP) problem as follows:
 \begin{align} \label{sub_xi2}
     & \max\limits_{\left\{  {\bm \Xi},  { \mu},  v \right\}}
  v,  \\
 &\operatorname{s.t.}  \begin{array}[t]{llll}
                \begin{bmatrix}
 \widetilde{Q}     -\left(2^{r}-1\right) \sigma_0^2-\mu-v  & \widetilde{\bm{\phi}}^{T}  \\
\widetilde{\bm{\phi}}  & \widetilde{\bm{\Phi}} +\mu \frac{\hat{d}_{\text{I2U}}^{2} }{\Upsilon^2} {\bf I}_3
\end{bmatrix} \succeq {\bf 0},  \\
           \left[\bm \Xi\right]_{ii}=1, i=1,\ldots,M, \\
           \text{rank}\left(  {\bm \Xi}\right)=1, \
           {\bm \Xi} \in {\mathbb{S}}_{+}^{M}, \ {v} \ge 0, { \mu} \ge 0,
  \end{array} \nonumber
\end{align}
where
$
 \widetilde{Q}     ={\bf 1}_{M}^{T} {\bm \Pi}     {\bm \Xi}
{\bm \Pi}    ^H {\bf 1}_{M}$,
$ [ \widetilde{\bm \phi}     ]_{q}=\text{tr} ({\bf D}_{   q} {\bm \Pi}     {\bm \Xi} {\bm \Pi}    ^{H} )
$,
$  [\widetilde{\bm \Phi}    ]_{sl}=
 \text{tr} ({\bf A}_{   sl} {\bm \Pi}     {\bm \Xi} {\bm \Pi}    ^{H}  )
$,
and  ${\bm \Pi}    =\text{diag}(  \text{diag}(\hat{\bf g}_{\text{I2U}}) { \bm \Theta } {\bf G}_{\text{B2I} } {\bf w})$.


To handle the  non-convexity
of  the rank constraint in (\ref{sub_xi2}), we adopt the SDR technique and
drop the rank constraint.
 \begin{align} \label{sub_xi3}
    & \max\limits_{\left\{  {\bm \Xi},  { \mu},  v \right\}}
  v,  \\
& \operatorname{s.t.}  \begin{array}[t]{llll}
                \begin{bmatrix}
 \widetilde{Q}     -\left(2^{r}-1\right) \sigma_0^2-\mu-v  & \widetilde{\bm{\phi}}^{T}  \\
\widetilde{\bm{\phi}}  & \widetilde{\bm{\Phi}} +\mu \frac{\hat{d}_{\text{I2U}}^{2} }{\Upsilon^2} {\bf I}_3
\end{bmatrix} \succeq {\bf 0},  \\
           \left[\bm \Xi\right]_{ii}=1, i=1,\ldots,M,
           \  {\bm \Xi} \in {\mathbb{S}}_{+}^{M}, {v} \ge 0, { \mu} \ge 0.
           \end{array} \nonumber
\end{align}

As such, the resulting problem becomes a convex SDP problem, which can be efficiently solved by CVX.
Then, we apply the Gaussian randomization  to obtain a feasible solution to
problem (\ref{sub_xi1}) based on the higher-rank solution obtained by solving (\ref{sub_xi3}).

\subsection{Alternating Optimization of ${\bf w}$ and $\bm \xi$}

Problem (\ref{op1}) is tackled by solving two sub-problems (\ref{sub_op2}) and (\ref{sub_xi3}) in an iterative manner, the details of which are summarized in Algorithm \ref{Alg}.
 \begin{algorithm}[!t]
\caption{Alternating Optimization Algorithm}
\label{Alg}
\begin{algorithmic}[1]
\State $\mathbf{Initialization:}$ Given feasible initial solutions ${\bm \xi}_{0}$, ${\bf w}_{0}$ and the iteration index $i=0$.

\Repeat

\State For given phase shifts ${\bm \xi}_{i}$,  solve problem (\ref{sub_op2}) and  then invoke Gaussian randomization to obtain a feasible beamforming vector ${\bf w}_{i+1}$.
\State For given beamforming vector ${\bf w}_{i+1}$, solve problem (\ref{sub_xi3}) and then exploit Guassian randomization to get a feasible phase shift vector ${\bm \xi}_{i+1}$.

\State $i \leftarrow i+1$.

\Until The fractional decrease of the objective value  in (\ref{sub_op2}) is below a threshold $\varepsilon > 0$.

\State $\mathbf{Output:}$ ${\bm \xi}^{\star} ={\bm \xi}_i$ and ${\bf w}^{\star}={\bf w}_{i}$.

  \end{algorithmic}
\end{algorithm}
\begin{remark}
Note that all considered sub-problems are SDP problems, which can be solved by interior point method.
Therefore, the approximate complexity of problem (\ref{sub_op2}) is $ {o}_1=\mathcal{O}(\sqrt{5}   ((\frac{N(N+1)}{2}+2)^3+16(\frac{N(N+1)}{2}+2)^2 +64(\frac{N(N+1)}{2}+2) ) )$, and that of problem (\ref{sub_xi3}) is $ {o}_2=\mathcal{O}( \sqrt{5}    ((\frac{M(M-1)}{2}+2)^3+16(\frac{M(M-1)}{2}+2)^2+64(\frac{M(M-1)}{2}+2) ) )$. Finally, the approximate complexity of Algorithm 2 per iteration is $ {o}_1+ {o}_2$.
With such  complexity, it is required that the BS has strong computational capacity so that the transmit beam and IRS phase shifts can be computed in real time.
\end{remark}
\begin{remark}
Since a SDR technique followed by Gaussian randomization is adopted when solving the two sub-problems,  the strict convergence of the proposed alternating algorithm can not be  guaranteed. But for each sub-problem,  such a
SDR approach followed by a sufficiently large number of randomizations guarantees at least a
$\frac{\pi}{4}$-approximation of the optimal objective value.
\end{remark}


\section{Simulation results}
In this section, we provide numerical results to evaluate the performance of the proposed algorithm.
 The considered system is assumed to operate at $28$ GHz with bandwidth of 100 MHz   and noise spectral power density of -169 dBm/Hz.
 The IRS is at the origin $(0, 0, 0)$ of a Cartesian coordinate system. The locations of the BS and the user are $(100, -100, 0) $ and $(20, 20, -20) $, respectively.
Unless otherwise specified, the following setup is used: $N=16$, $M=100$  and $\Upsilon=4$.\footnote{As pointed out in \cite{LIS2}, synchronously operating a large number of phase shifters  is a non-trivial task. The IRS may suffer from phase errors. It have been shown in  \cite{LIS1}  that the average received SNR is attenuated by $\phi^2$ where $\phi \triangleq \mathbb{E}\left\{e^{j \theta} \right\}$ with $\theta$ being the phase error.  One possible way to reduce the phase error is to use phase shifters with adaptive tuning capability. }

Fig. \ref{f1} compares the proposed robust scheme with a “Non-robust” optimization scheme \cite{wu2019intelligent}.
It is worth noting that the target rate $r$ is a predefined threshold, while the x-axis in Fig.2 ($R$) is the rate achieved under different location errors of each channel realization.
The proposed robust beamforming scheme requires that the achievable  rate  under all possible location errors should exceed the threshold $r$, while the non-robust beamforming scheme regards the estimated
channels as perfect channels, and aims to maximize the achievable rate without location uncertainty, using the same power as the proposed beamforming scheme.
As can be readily seen, the variation of the user rate with robust optimization is much smaller than that with non-robust optimization. Moreover, the target rate can  be always guaranteed regardless of user location errors. By contrast, with non-robust optimization,  the user rate varies over a wide range. For instance, with $r=6$, the user rate ranges  from near 0 bits/s/Hz to about 7.2 bits/s/Hz. Moreover, there is a high probability (over $80\%$) that the target rate  can  not  be  achieved.

\begin{figure}[!ht]
  \centering
  \includegraphics[width=3 in]{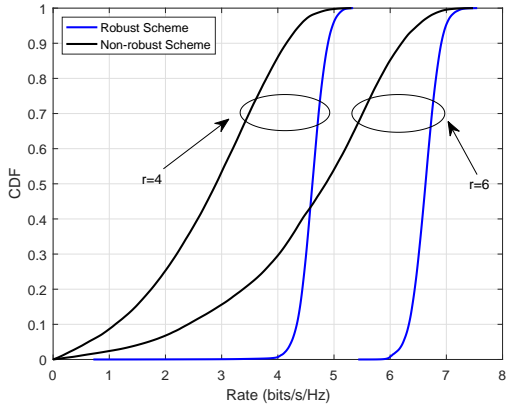}
   \caption{ Comparison of the proposed robust optimization and the non-robust optimization.}
  \label{f1}
\end{figure}

Fig. \ref{f2} presents the transmit power versus the target rate with different location uncertainty and numbers of reflecting elements. Obviously, to achieve a higher target rate, more power is required. Also, as location uncertainty (measured by $\Upsilon$) increases, the required transmit power becomes larger so that the target rate can be achieved for all possible location errors. Besides, since both the  beamforming gain and the  aperture gain of the IRS grow with the number of reflecting elements, the transmit power drops significantly with the number of reflecting elements, implying the benefit of the IRS in power saving.

\begin{figure}[!ht]
  \centering
  \includegraphics[width=3 in]{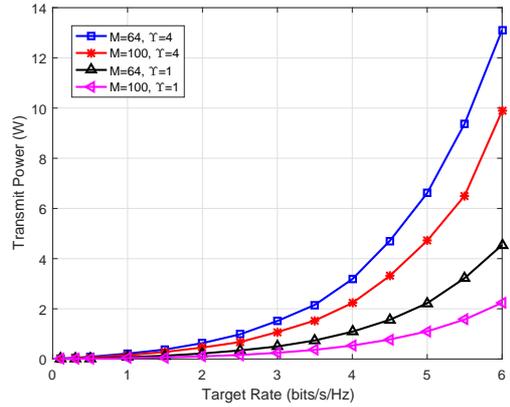}
   \caption{ Transmit power versus the target rate.}
  \label{f2}
\end{figure}

\section{Conclusion}
In this paper, considering user location uncertainty, we study the robust beamforming design for an IRS-aided communication system. We first handle the location uncertainty  by exploiting techniques of Taylor expansion and S-Procedure. Then SDR is used to transform the non-convex problem into a sequence of SDP sub-problems, which can be efficiently solved via some optimization tools, for example, CVX.  Simulation results demonstrate that with our proposed robust algorithm, the QoS requirement is always  met regardless of location uncertainty, while
with a non-robust optimization algorithm \cite{wu2019intelligent} which has the same transmit power as the proposed algorithm,  the QoS requirement is met   with a  probability below $20 \% $.

\nocite{*}
\bibliographystyle{IEEE}
 \begin{footnotesize}

\end{footnotesize}

\end{document}